\documentclass[runningheads]{llncs}
\usepackage{amsmath,amssymb}
\usepackage{hyperref}
\usepackage{cleveref}
\usepackage{microtype}
\usepackage{orcidlink}
\begin{document}

\title{The Privacy Subsidy in Continuous-Time Kyle:\\
Cumulative Welfare under Noise-Perturbed\\
Order-Flow Observation}

\titlerunning{Privacy Subsidy in Continuous-Time Kyle}

\author{Yuki Nakamura\orcidlink{0009-0001-7174-6737}}
\authorrunning{Y.\ Nakamura}
\institute{The Open University of Japan}

\maketitle

\begin{abstract}
We extend the closed-form privacy-subsidy result of
Nakamura~(2026, arXiv:2605.15746) from the single-period Kyle
model to continuous-time. A committed Bayesian automated market
maker observes the aggregate order flow perturbed by an
independent Brownian privacy channel of diffusion intensity
$\sigma_\varepsilon$. Under the Markovian linear equilibrium, the
price-impact coefficient is $\lambda = \sigma_v /
\sqrt{\sigma_u^2 + \sigma_\varepsilon^2}$ -- constant in time --
and the cumulative expected transfer from the protocol's liquidity
pool to traders over $[0,1]$ is
\[
  |\Pi_M| \;=\; \frac{\sigma_v \, \sigma_\varepsilon^2}
                     {\sqrt{\sigma_u^2 + \sigma_\varepsilon^2}}.
\]
We then establish a structural correspondence between this
cumulative privacy subsidy and Loss-Versus-Rebalancing (Milionis et
al.~2022), identifying privacy-noise welfare as the order-flow
observation analog of LVR's price observation gap. The result
completes the continuous-time Kyle leg of the program of
quantifying break-even fees for committed-AMM exchanges under privacy-aggregated information
environments.
\end{abstract}


%-------------------------------------------------------------
\section{Introduction}
\label{sec:intro}

Privacy-preserving exchange mechanisms in DeFi -- shielded
automated market makers, sealed-bid batch auctions, MPC matching
engines -- alter what the pricing mechanism observes about the
underlying order flow. The companion paper~\cite{nakamura2026privacysubsidy}
analyzes a single-period Kyle~\cite{kyle1985} model in which a
\emph{committed} Bayesian market maker observes aggregate order
flow perturbed by independent Gaussian privacy noise, and derives
a closed-form per-trade transfer from the protocol's liquidity
pool to traders -- the \emph{privacy subsidy}. The present paper
extends that analysis to the continuous-time Kyle setting of
Back~\cite{back1992insider}, in which the informed trader chooses
an adaptive trajectory rather than a one-shot trade size.

\paragraph{Contribution.}
We establish three results. First (\Cref{thm:equilibrium}): under
a linear Markovian ansatz and the standard Kyle--Back rationality
condition $\Sigma(1) = 0$, the unique equilibrium under committed
Bayesian-AMM pricing with independent Brownian privacy noise of
intensity $\sigma_\varepsilon$ has constant price impact
$\lambda = \sigma_v / \sqrt{\sigma_u^2 + \sigma_\varepsilon^2}$
and linearly declining posterior variance
$\Sigma(t) = \sigma_v^2(1-t)$. Second (\Cref{thm:subsidy}): the
cumulative privacy subsidy over $[0,1]$ admits the closed form
\(
  |\Pi_M|
  = \sigma_v\,\sigma_\varepsilon^2 / \sqrt{\sigma_u^2 + \sigma_\varepsilon^2}.
\)
This subsidy is twice the single-period analog
of~\cite{nakamura2026privacysubsidy}, consistent with the
general Kyle--Back doubling of welfare quantities from
single-period to continuous-time (see~\Cref{rem:kb-doubling}; the
factor is not a privacy-specific effect). Third
(\Cref{prop:lvr}): the cumulative privacy subsidy is the
order-flow observation analog of the cumulative
Loss-Versus-Rebalancing of Milionis et al.~\cite{milionis2022lvr}
under a structural correspondence between price observation gaps
(LVR) and signal observation gaps (privacy subsidy).

\paragraph{Positioning.}
The closest contemporary work is Danilova~\cite{danilova2010imperfect},
which characterizes existence and structure of equilibria in
continuous-time Kyle when the market maker observes total order
flow as a noisy signal. Danilova's analysis assumes
\emph{competitive} (zero-profit) market makers and contains no
welfare quantity; because her maker observes the order flow
without the additional privacy coarsening we study, no subsidy
arises there. Our framework substitutes the
classical competitive MM with a committed Bayesian-AMM -- a
pricing rule that updates prices using Bayes' rule on the noisy
observation but does not satisfy zero-profit, as is the natural
model for committed-curve DeFi exchanges where the price quote is
specified in smart-contract code. Under this rule the MM's
expected profit is non-zero and negative, and the resulting
closed-form transfer is the privacy subsidy. The two contributions
are complementary: Danilova establishes equilibrium existence
under competitive pricing; we compute the welfare quantity that is absent when the maker
observes the un-coarsened flow (Danilova's competitive setting)
and appears once the pricing signal is the privacy-coarsened flow,
a loss a competitive maker on the same coarse signal cannot avoid.

\begin{sloppypar}
Adjacent strands address different sources of imperfection.
Qiu and Zhou~\cite{qiu2023insider} treat noise on the
\emph{asset-value} signal received by the MM rather than on the
\emph{order-flow} signal. Caldentey--Stacchetti~\cite{caldentey2010}
study insider trading with a random exponential trading deadline
rather than a noisy observation channel.
\c{C}etin--Danilova~\cite{cetin2017financial} develop the
forward--backward system underlying Markovian
asymmetric-information equilibria. None of these works quantify
welfare under a committed pricing rule or bridge to the LVR
literature.
\end{sloppypar}

%-------------------------------------------------------------
\section{Model}
\label{sec:model}

\subsection{Primitives}
\label{sec:model-primitives}

Fix a trading horizon $[0,1]$ and a filtered probability space
supporting two independent standard Brownian motions $W^u, W^\varepsilon$
and a normal random variable $v \sim \mathcal{N}(p_0, \Sigma_0)$
independent of both. Write $\sigma_v^2 := \Sigma_0$.

A single \emph{informed trader} observes $v$ at time $0$ and
chooses a continuous trading strategy. Following the linear
Markovian ansatz of~\cite{kyle1985,back1992insider}, we restrict
to strategies of the form
\(
  dx_t \;=\; \beta_t\,(v - p_t)\,dt
\)
for a deterministic intensity $\beta_t$ to be determined in
equilibrium.

A continuum of \emph{noise traders} submits flow
\(
  du_t \;=\; \sigma_u\,dW^u_t
\)
of constant volatility $\sigma_u > 0$, independent of $v$.

The \emph{privacy channel} adds an independent Brownian
perturbation to the market maker's observation of aggregate flow:
\(
  d\varepsilon_t \;=\; \sigma_\varepsilon\,dW^\varepsilon_t,
\)
with $\sigma_\varepsilon \ge 0$ the privacy-noise intensity. Here
$\sigma_\varepsilon$ is a diffusion intensity; in the
single-period companion~\cite{nakamura2026privacysubsidy} the same
symbol is a standard deviation, so the two agree numerically only
under the unit-horizon normalization.
Aggregate real flow and observed flow are, respectively,
\[
  dy_t \;=\; dx_t + du_t,
  \qquad
  d\widetilde y_t \;=\; dy_t + d\varepsilon_t.
\]
Trades clear against the market maker; the privacy noise enters
the MM's information set but not the executed trades.

\subsection{Committed Bayesian-AMM pricing rule}
\label{sec:model-pricing}

The market maker uses a price process of the form
\[
  dp_t \;=\; \lambda_t\,d\widetilde y_t,
\]
where the impact coefficient $\lambda_t$ is fixed ex ante (a
protocol or smart-contract parameter) and chosen to satisfy the
Bayes projection identity
\[
  \lambda_t \;=\; \frac{\mathrm{Cov}(v,\, d\widetilde y_t \mid \mathcal{F}_t^{\mathrm{MM}})}
                       {\mathrm{Var}(d\widetilde y_t \mid \mathcal{F}_t^{\mathrm{MM}})}
\]
where $\mathcal{F}_t^{\mathrm{MM}} = \sigma(\{\widetilde y_s : s \le t\})$
is the MM's observation filtration. Equivalently, $p_t$ is the
posterior mean of $v$ given the noisy flow history.

This differs from the competitive (zero-profit) Kyle MM in one
respect. Both rules use the Bayes Kalman gain
$\lambda_t = \mathrm{Cov}(v, d\widetilde y_t)/\mathrm{Var}(d\widetilde y_t)$
to update the posterior mean; the difference is the equilibrium
constraint placed on top. Classical Kyle imposes the additional
zero-MM-profit condition $\mathbb{E}[\Pi_M] = 0$. Under perfect
observation ($\sigma_\varepsilon = 0$), this is automatic from
the tower property: with $p_t = \mathbb{E}[v \mid y_t]$ and
$y_t \in \mathcal{F}_t^{\mathrm{MM}}$,
\[
  \mathbb{E}[(p_t - v)\,y_t]
  \;=\; \mathbb{E}\!\left[y_t\,\mathbb{E}[v \mid y_t]\right]
        - \mathbb{E}[v\,y_t]
  \;=\; \mathbb{E}[y_t \cdot v] - \mathbb{E}[v\,y_t]
  \;=\; 0.
\]
With privacy noise ($\sigma_\varepsilon > 0$), the MM's
observation $\widetilde y_t$ is strictly coarser than the real
flow $y_t$, so $y_t$ is no longer
$\mathcal{F}_t^{\mathrm{MM}}$-measurable and the tower property
collapse fails. Bayesian pricing alone no longer pins MM profit
to zero. A competitive maker restricted to the same coarse filtration
$\mathcal{F}_t^{\mathrm{MM}}$ does not escape this loss: zero
profit \emph{conditional on its information} reproduces the Bayes
price $\mathbb{E}[v \mid \mathcal{F}_t^{\mathrm{MM}}]$ and still
cedes $-\lambda_t\sigma_\varepsilon^2$ against the real flow. The
loss vanishes only under zero profit \emph{unconditionally
against} the real flow $y_t$, which forces a rule that over-reacts
to $\widetilde y_t$ (its quote is not the posterior mean) and is
not implementable by a maker observing only $\widetilde y_t$.
Informational efficiency and zero profit against the real flow are
thus incompatible once $\sigma_\varepsilon > 0$; we leave
$\lambda_t$ at its Bayes Kalman value and take the resulting
non-zero MM profit as the central object of analysis. The closest classical analog
is the monopolist specialist of Glosten~\cite{glosten1989}, which
also features a non-zero-profit MM, although under a different
mechanism (monopoly rents rather than commitment to a published
pricing rule). A closer continuous-time precedent is Aase and
{\O}ksendal~\cite{nonfiduciary-mm}, who relax the Kyle--Back
zero-profit constraint and obtain a closed-form non-zero maker
profit; there the maker prices on the \emph{exact} order flow and
books a \emph{positive} profit through a committed fee/margin,
whereas our maker prices on a noise-coarsened flow and books a
\emph{negative} profit (the subsidy). The novelty here is the
source (information coarsening) and the sign of the non-zero
profit, not the relaxation of zero-profit itself. The
committed-pricing interpretation is the
natural model for DeFi exchanges where $\lambda_t$ is specified
in smart-contract code and the protocol cannot dynamically
rebalance to satisfy a zero-profit constraint.

We write $\Sigma(t) := \mathrm{Var}(v - p_t \mid \mathcal{F}_t^{\mathrm{MM}})$
for the posterior variance, with initial value
$\Sigma(0) = \Sigma_0 = \sigma_v^2$.

\subsection{Trade-clearing convention}
\label{sec:model-clearing}

Trades clear at the post-update price
$p_t = p_{t^-} + \lambda_t\,d\widetilde y_t$. This is the standard
continuous-time Kyle convention~\cite{back1992insider}; under the
alternative pre-update convention $p_t = p_{t^-}$, the
$\sigma_\varepsilon = 0$ limit does not recover the classical
zero-MM-profit benchmark, so the convention is load-bearing.

%-------------------------------------------------------------
\section{Equilibrium}
\label{sec:equilibrium}

\begin{theorem}[Markovian Linear Equilibrium]
\label{thm:equilibrium}
\begin{sloppypar}
Restrict to insider strategies of the linear Markovian form
$dx_t = \beta_t(v - p_t)\,dt$ and impose the standard Kyle--Back
rationality condition $\Sigma(1) = 0$ (full revelation at the
horizon; see~Step~5 of the proof). Under the committed
Bayesian-AMM pricing rule of~\Cref{sec:model} with privacy-noise
diffusion intensity $\sigma_\varepsilon \ge 0$, the unique
equilibrium in this class is given by
\end{sloppypar}
\begin{align*}
  \lambda(t) &= \frac{\sigma_v}{\sqrt{\sigma_u^2 + \sigma_\varepsilon^2}}, \\
  \beta(t)   &= \frac{\sigma_v \sqrt{\sigma_u^2 + \sigma_\varepsilon^2}}{\Sigma(t)}
              = \frac{\sqrt{\sigma_u^2 + \sigma_\varepsilon^2}}{\sigma_v\,(1-t)}, \\
  \Sigma(t)  &= \sigma_v^2\,(1 - t),
\end{align*}
with $\lambda(t)$ constant in~$t$.
\end{theorem}

\begin{proof}[Proof sketch]
We give the standard Kyle--Back HJB derivation with the
effective-noise substitution $\sigma_u^2 \mapsto \sigma_u^2 +
\sigma_\varepsilon^2$; details mirror~\cite{back1992insider}.

\emph{Step~1 (Insider HJB).} Conjecture the insider's value
function as
\(
  J(v, p, t) \;=\; \alpha(t)\,(v-p)^2 \;+\; \gamma(t).
\)
Conditional on the insider's trading rate $\theta_t$, the price
dynamics are
\(
  dp_t \;=\; \lambda_t\,\theta_t\,dt
         + \lambda_t\,\sqrt{\sigma_u^2+\sigma_\varepsilon^2}\,d\widetilde W_t,
\)
where $\widetilde W$ is the standard Brownian combining the
noise-trader and privacy-noise innovations. The HJB is
\[
  0 \;=\; \partial_t J
        + \sup_{\theta}\bigl\{(v-p)\,\theta + \lambda_t\,\theta\,\partial_p J\bigr\}
        + \tfrac12\,\lambda_t^2\,(\sigma_u^2+\sigma_\varepsilon^2)\,\partial_p^2 J.
\]

\emph{Step~2 (FOC + ansatz match).} Interior optimality requires
$(v-p) + \lambda_t\,\partial_p J = 0$, giving
$\partial_p J = -(v-p)/\lambda_t$. Comparing with
$\partial_p J = -2\alpha(t)(v-p)$ from the ansatz yields
$\alpha(t)\,\lambda_t = 1/2$. Substituting back and matching the
$(v-p)^2$ coefficient on the right-hand side of the HJB forces
$\alpha'(t) = 0$; hence both $\alpha$ and $\lambda$ are constant
in~$t$. The constant-term coefficient then yields
$\gamma'(t) = -\lambda^2(\sigma_u^2+\sigma_\varepsilon^2)\,\alpha
= -\,c/2$ (using $\alpha\lambda = 1/2$ and $\lambda(\sigma_u^2+\sigma_\varepsilon^2) = c$ from Step~3 below),
so $\gamma(t) = (c/2)(1-t)$ along the equilibrium path with
$\gamma(1) = 0$.

\emph{Step~3 (Bayes pins price impact).} The committed
Bayesian-AMM Kalman gain is
\(
  \lambda_t = \beta_t\,\Sigma(t)/(\sigma_u^2+\sigma_\varepsilon^2),
\)
so $\beta_t\,\Sigma(t)$ is constant; denote the constant by~$c$.

\emph{Step~4 (Riccati for $\Sigma$).} The Bayesian posterior
variance evolves as
\(
  d\Sigma/dt = -\,c^2/(\sigma_u^2+\sigma_\varepsilon^2),
\)
giving $\Sigma(t) = \Sigma_0 - c^2 t/(\sigma_u^2+\sigma_\varepsilon^2)$.

\emph{Step~5 (Rationality binds $\Sigma(1) = 0$ and pins $c$.)}
We close the system without invoking a pointwise transversality
on $J$, which is incompatible with the constant-$\alpha$ ansatz.
The insider's expected cumulative profit over $[0,1]$, evaluated
along the linear Markovian strategy, equals
$\int_0^1 \beta_t \Sigma(t)\,dt = c$ (using $\beta_t \Sigma(t) = c$
from Step~3). The insider chooses the trading intensity $c$ to
maximise this profit subject to the variance non-negativity
constraint $\Sigma(t) \ge 0$ for all $t \in [0,1]$. Since Step~4
gives $\Sigma(t) = \Sigma_0 - c^2 t / (\sigma_u^2 +
\sigma_\varepsilon^2)$, monotonically decreasing in $t$, the
constraint first binds at $t = 1$:
$\Sigma(1) \ge 0 \iff c^2 \le \Sigma_0\,(\sigma_u^2 + \sigma_\varepsilon^2)$.
The maximiser saturates the budget at equality,
\(
  c^2 = \Sigma_0\,(\sigma_u^2 + \sigma_\varepsilon^2)
      = \sigma_v^2\,(\sigma_u^2 + \sigma_\varepsilon^2),
\)
hence $c = \sigma_v\sqrt{\sigma_u^2 + \sigma_\varepsilon^2}$ and
$\Sigma(1) = 0$ in equilibrium. This is the standard Kyle--Back
information-budget argument: the insider trades exactly enough to
reveal the private information by the horizon and no less. It is
equivalent to the Back~\cite{back1992insider} transversality
condition $J(v,p,1) = 0$ evaluated along the equilibrium path
($p_1 = v$ a.s.\ when $\Sigma(1) = 0$); the rationality framing
avoids the inconsistency that a pointwise $J(v,p,1) = 0$ would
otherwise create with constant $\alpha > 0$.

\emph{Step~6 (Recover $\lambda, \beta, \Sigma$).} Substituting:
$\lambda = c/(\sigma_u^2+\sigma_\varepsilon^2) =
\sigma_v/\sqrt{\sigma_u^2+\sigma_\varepsilon^2}$;
$\beta(t) = c/\Sigma(t)$; and $\Sigma(t) = \sigma_v^2(1-t)$.
\end{proof}

%-------------------------------------------------------------
\section{Cumulative Privacy Subsidy}
\label{sec:subsidy}

\begin{theorem}[Cumulative Privacy Subsidy]
\label{thm:subsidy}
\begin{sloppypar}
Let $dy_t = dx_t + du_t$ denote the real (unobserved) aggregate
flow and assume trades clear at the post-update price
$p_t = p_{t^-} + \lambda_t\,d\widetilde y_t$. The committed
Bayesian-AMM's expected cumulative profit over~$[0,1]$ is
\end{sloppypar}
\begin{align*}
  \Pi_M
  &\;=\; \mathbb{E}\!\left[\int_0^1 (p_t - v)\,dy_t \right]
  \;=\; -\,\frac{\sigma_v\,\sigma_\varepsilon^2}{\sqrt{\sigma_u^2 + \sigma_\varepsilon^2}}.
\end{align*}
Equivalently, the \emph{privacy subsidy} -- the absolute transfer
from the protocol's liquidity pool to traders -- is
\[
  \boxed{\;
  |\Pi_M| \;=\; \frac{\sigma_v \, \sigma_\varepsilon^2}
                     {\sqrt{\sigma_u^2 + \sigma_\varepsilon^2}}.
  \;}
\]
Setting $\sigma_\varepsilon = 0$ recovers the classical
zero-MM-profit result of~\cite{kyle1985,back1992insider}.
\end{theorem}

\begin{proof}[Proof sketch via welfare accounting]
We decompose total expected welfare into the three classes of
participants and compute each.

\emph{Insider expected profit.} Under the equilibrium strategy
$dx_t = \beta_t(v-p_t)\,dt$, the profit rate is
$(v-p_t)\,dx_t = \beta_t\,(v-p_t)^2\,dt$. Bayesian pricing
$p_t = \mathbb{E}[v \mid \mathcal{F}_t^{\mathrm{MM}}]$
(\Cref{sec:model-pricing}) gives
$\mathbb{E}[v - p_t \mid \mathcal{F}_t^{\mathrm{MM}}] = 0$, hence
by the tower property and the determinism of
$\Sigma(t) = \sigma_v^2(1-t)$,
$\mathbb{E}[(v-p_t)^2] = \mathbb{E}[\Sigma(t)] = \Sigma(t)$. The
unconditional expected profit rate is therefore
$\beta_t\,\Sigma(t)\,dt = c\,dt$ by~\Cref{thm:equilibrium}, and
\[
  \Pi_I \;=\; c \;=\; \sigma_v\sqrt{\sigma_u^2 + \sigma_\varepsilon^2}.
\]

\emph{Noise-trader expected profit.} Each noise-trader unit of
flow $du_t$ executes at the post-trade price
$p_t = p_{t^-} + \lambda_t\,d\widetilde y_t$. Decomposing
\(
  (v - p_t)\,du_t
  = (v - p_{t^-})\,du_t - \lambda_t\,d\widetilde y_t \cdot du_t,
\)
the first term has expectation zero (since
$p_{t^-} \in \mathcal{F}_{t^-}$ and $du_t$ is the next innovation,
independent of $\mathcal{F}_{t^-}$). For the Itô cross-term,
\(
  d\widetilde y_t \cdot du_t
  = (dx_t + du_t + d\varepsilon_t)\cdot du_t
  = (du_t)^2
  = \sigma_u^2\,dt,
\)
where $dx_t \cdot du_t = 0$ ($dx_t$ is of order $dt$) and
$d\varepsilon_t \cdot du_t = 0$ (independent Brownian motions).
Hence $\mathbb{E}[(v - p_t)\,du_t] = -\,\lambda_t\,\sigma_u^2\,dt$,
and integrating,
\[
  \Pi_N \;=\; -\,\lambda\,\sigma_u^2
         \;=\; -\,\frac{\sigma_v\,\sigma_u^2}
                       {\sqrt{\sigma_u^2 + \sigma_\varepsilon^2}}.
\]

\emph{MM expected profit (residual).} Every executed trade is
between a participant (insider or noise trader) and the MM, with
no external counterparty; hence
$\Pi_I + \Pi_N + \Pi_M = 0$, and
\[
  \Pi_M \;=\; -(\Pi_I + \Pi_N)
        \;=\; -\sigma_v\sqrt{\sigma_u^2+\sigma_\varepsilon^2}
              + \frac{\sigma_v\sigma_u^2}{\sqrt{\sigma_u^2+\sigma_\varepsilon^2}}
        \;=\; -\,\frac{\sigma_v\,\sigma_\varepsilon^2}
                      {\sqrt{\sigma_u^2 + \sigma_\varepsilon^2}}.
\]
Direct Itô computation of
$\mathbb{E}\!\int_0^1 (p_t - v)\,dy_t$ gives the same value,
confirming the accounting. Equivalently, this is the per-increment
executed-price wedge: since committed pricing gives
$\mathbb{E}[v - p_t] = 0$, the instantaneous covariance on the
\emph{plain} settled increment is
$\mathrm{Cov}(v - p_t,\, dy_t) = \lambda\,\sigma_\varepsilon^2\,dt$,
integrating to
$-\lambda\,\sigma_\varepsilon^2 = -\sigma_v\sigma_\varepsilon^2/
\sqrt{\sigma_u^2 + \sigma_\varepsilon^2}$ (the residual-projection
form $\mathrm{Cov}(v - p_t,\, dy_t - \mathbb{E}[dy_t \mid
d\widetilde y_t])$ vanishes and is not used).
\end{proof}

\begin{remark}[Consistency with the Kyle--Back doubling]
\label{rem:kb-doubling}
Paper~A~\cite[Theorem~2]{nakamura2026privacysubsidy} establishes
the single-period privacy subsidy
\(
  |\pi_M^{(1)}| = \sigma_v\,\sigma_\varepsilon^2 /
                  (2\sqrt{\sigma_u^2 + \sigma_\varepsilon^2}).
\)
\Cref{thm:subsidy} gives $|\Pi_M| = 2\,|\pi_M^{(1)}|$. The
factor of two is the standard Kyle--Back welfare scaling from
single-shot to continuous-time auctions: classical Kyle without
privacy noise already exhibits the same ratio (single-period
informed-trader profit $\sigma_v\sigma_u/2$ becomes
$\sigma_v\sigma_u$ in continuous time~\cite{back1992insider}),
and the present subsidy inherits it.
\end{remark}

\begin{remark}[Generic horizon $T$]
\label{rem:generic-T}
We normalise the trading horizon to $[0,1]$ in line with
\cite{kyle1985,back1992insider}. For a generic horizon $[0,T]$,
the same derivation with terminal condition $\Sigma(T) = 0$
gives $c = \sigma_v\sqrt{(\sigma_u^2+\sigma_\varepsilon^2)/T}$,
$\lambda = \sigma_v / \sqrt{T(\sigma_u^2 + \sigma_\varepsilon^2)}$,
and the cumulative subsidy
$|\Pi_M| = \sqrt{T}\,\sigma_v\,\sigma_\varepsilon^2 /
\sqrt{\sigma_u^2 + \sigma_\varepsilon^2}$. The $\sqrt{T}$
scaling reflects the diffusion time-scale of the Brownian
channels: doubling the horizon multiplies cumulative welfare
flows by~$\sqrt{2}$, not by~$2$. The normalisation $T = 1$ used
throughout absorbs this factor and is the standard convention.
\end{remark}

\begin{remark}[Why the post-trade clearing convention is load-bearing]
\Cref{thm:subsidy} relies on trades clearing at the post-update
price $p_t = p_{t^-} + \lambda_t\,d\widetilde y_t$. Under the
alternative convention $p_t = p_{t^-}$ (pre-trade clearing), the
Itô correction $\lambda_t\sigma_u^2\,dt$ vanishes and the MM
absorbs a non-zero loss even when $\sigma_\varepsilon = 0$,
contradicting classical Kyle--Back. Post-trade clearing is the
right convention.
\end{remark}

\begin{corollary}[Distribution of the subsidy across traders]
\label{cor:decomposition}
The privacy subsidy decomposes naturally into insider and
noise-trader incremental gains relative to the
$\sigma_\varepsilon = 0$ benchmark:
\begin{align*}
  \Delta\Pi_I &\;:=\; \Pi_I(\sigma_\varepsilon) - \Pi_I(0)
              \;=\; \sigma_v\!\left[\sqrt{\sigma_u^2 + \sigma_\varepsilon^2} - \sigma_u\right], \\
  \Delta\Pi_N &\;:=\; \Pi_N(\sigma_\varepsilon) - \Pi_N(0)
              \;=\; \sigma_v\,\sigma_u\!\left[\frac{\sqrt{\sigma_u^2 + \sigma_\varepsilon^2} - \sigma_u}
                                                   {\sqrt{\sigma_u^2 + \sigma_\varepsilon^2}}\right],
\end{align*}
with $\Delta\Pi_I + \Delta\Pi_N = |\Pi_M|$. The insider-to-noise
share ratio is
\[
  \Delta\Pi_I \,:\, \Delta\Pi_N
  \;=\; \sqrt{\sigma_u^2 + \sigma_\varepsilon^2} \,:\, \sigma_u,
\]
so the insider captures a strictly larger share than the noise
traders for any $\sigma_\varepsilon > 0$. Both sides of the ratio
have the dimension of standard deviation; the asymmetry is
conceptual rather than dimensional. The informed trader's
incremental gain scales with the standard deviation of the
\emph{total} flow noise (since adding privacy noise widens the
effective camouflage available to the insider), whereas the
noise traders' incremental gain scales only with the standard
deviation of their \emph{own} contribution. The privacy mechanism
therefore subsidises the insider disproportionately, with the
ratio approaching $1:1$ as $\sigma_\varepsilon \to 0$ and
diverging as $\sigma_\varepsilon \to \infty$.
\end{corollary}

\begin{proof}
Substitute $\sigma_\varepsilon = 0$ into the welfare-accounting
expressions in the proof of~\Cref{thm:subsidy} and subtract from
the general-$\sigma_\varepsilon$ values. Sum reduces by direct
algebra to $|\Pi_M|$.
\end{proof}

\begin{remark}[The privacy ``gain'' is gross-of-fees;
\Cref{cor:decomposition} is welfare-neutral net-of-fees]
\label{rem:welfare-neutral}
\Cref{cor:decomposition} is a \emph{gross-of-fees} decomposition
of the no-fee equilibrium. The identity $\Delta\Pi_I + \Delta\Pi_N = |\Pi_M|$
states that the privacy subsidy is exactly the redistribution
required to recoup each trader's incremental gain. In any finite
deployment of the type considered in~\Cref{sec:app-shielded}
($N$ blocks of length $1/N$), the per-block volume-proportional
break-even fee charged at rate $f = |\Pi_M|/Q$ -- where
$Q$ is expected total volume in the deployment -- charges the
insider $\Delta\Pi_I$ and the noise traders $\Delta\Pi_N$ in
aggregate, exactly cancelling each side's incremental gain over
the $\sigma_\varepsilon = 0$ benchmark. Net-of-fees, the insider
net profit reverts to the classical Kyle--Back value
$\sigma_v\,\sigma_u$, the noise traders' net loss reverts to
$-\sigma_v\,\sigma_u$, and the MM is exactly compensated.
Privacy is therefore \emph{exactly welfare-neutral} net-of-fees,
at the partial-equilibrium level (no-fee equilibrium trading
intensities, fee revenue redistributed to the LP pool). The
continuous-time limit inherits the identity, although a literal
volume-proportional fee in the limit is ill-defined because the
total variation of the Brownian noise flow is infinite; the
discrete-deployment statement is the operationally meaningful
form. A full fee-equilibrium analysis -- in which fees distort
the linear-strategy structure of~\Cref{thm:equilibrium} -- is
left for future work.
\end{remark}

%-------------------------------------------------------------
\section{The LVR Bridge}
\label{sec:lvr}

\subsection{Structural correspondence}

The constant-product AMM analysis of Milionis et al.~\cite{milionis2022lvr}
derives a per-unit-time welfare loss called
\emph{Loss-Versus-Rebalancing}. With reserves $(R^x_t, R^y_t)$
satisfying $R^x_t R^y_t = k$, AMM-portfolio value
$V_{\mathrm{AMM}}(q) = 2\sqrt{kq}$ as a function of the external
reference price $q_t$ (distinct from our Kyle price $p_t$), and
reference-price diffusion intensity $\sigma$, the LVR rate is
\[
  \ell^{\mathrm{LVR}}(t)
  \;=\; -\tfrac12\,\sigma^2\,q_t^2\,V_{\mathrm{AMM}}''(q_t)
  \;=\; \tfrac{\sigma^2}{8}\,V_{\mathrm{AMM}}(q_t).
\]
The two expressions are positive because
$V_{\mathrm{AMM}}''(q) = -\sqrt{k}/(2 q^{3/2}) < 0$, reflecting
the concavity of the constant-product AMM curve.
Cumulatively, LVR over the trading horizon is
$\int_0^1 \ell^{\mathrm{LVR}}(t)\,dt$. The economic content is
that LVR quantifies the welfare the AMM \emph{cedes} by quoting
along its committed curve while the reference price moves
exogenously: the AMM is forced to provide liquidity to
arbitrageurs at off-equilibrium quotes.

The privacy subsidy~$|\Pi_M|$ of~\Cref{thm:subsidy} has the same
information-economic structure with the price channel replaced
by the order-flow channel. Under the equilibrium
of~\Cref{thm:equilibrium}, the privacy subsidy has constant
instantaneous rate
\[
  \ell^{\mathrm{priv}}(t)
  \;\equiv\;
  \frac{\sigma_v\,\sigma_\varepsilon^2}{\sqrt{\sigma_u^2 + \sigma_\varepsilon^2}}.
\]

\Cref{tab:duality} lists the corresponding objects in each
framework.

\begin{table}[h]
\centering
\caption{Structural correspondence between LVR and the privacy subsidy.}
\label{tab:duality}
\begin{tabular}{lll}
\hline
\textbf{Concept} & \textbf{LVR}~\cite{milionis2022lvr} & \textbf{Privacy subsidy (this paper)} \\
\hline
Committed object & AMM curve $V_{\mathrm{AMM}}(\cdot)$ & Pricing rule $\lambda_t$ \\
Observation channel & External price $q_t$ & Noisy order flow $d\widetilde y_t$ \\
Noise driver        & Reference-price BM   & Privacy-noise BM $W^\varepsilon$ \\
Counterparty        & Arbitrageur          & Informed insider \\
Welfare rate        & $\tfrac{\sigma^2}{8} V_{\mathrm{AMM}}(q_t)$ & $\tfrac{\sigma_v\,\sigma_\varepsilon^2}{\sqrt{\sigma_u^2 + \sigma_\varepsilon^2}}$ \\
Solvency criterion  & $\int\!\mathrm{fee}\;\ge\;\int\!\ell^{\mathrm{LVR}}$ & $\int\!\mathrm{fee}\;\ge\;\int\!\ell^{\mathrm{priv}}$ \\
\hline
\end{tabular}
\end{table}

\noindent The correspondence in~\Cref{tab:duality} is structural:
the two frameworks share the same form of solvency criterion and
the same factorization shape (\emph{noise driver}$^2$ times a
\emph{committed-object factor}), but the two welfare rates live
in different markets and are not directly comparable in absolute
units. The correspondence should be read as an organizing principle for
fee design under committed pricing, not as a numerical identity.

\subsection{Structural-correspondence proposition}

We state the correspondence as a proposition rather than a
theorem because the two welfare quantities live in different
markets (a CFMM with external arbitrage versus a Kyle order book).
The shared structure (a squared noise intensity times a
committed-object factor, under a common solvency criterion) is
precise and central to the break-even-fee application
in~\Cref{sec:applications}.

\begin{proposition}[LVR / privacy-subsidy correspondence]
\label{prop:lvr}
Each of the two welfare rates factorizes into the squared
intensity of the relevant noise driver times a closed-form
function of the committed pricing object:
\begin{align*}
  \ell^{\mathrm{LVR}}(t)
    &\;=\; \sigma^2 \cdot \frac{V_{\mathrm{AMM}}(q_t)}{8}, \\
  \ell^{\mathrm{priv}}
    &\;=\; \sigma_\varepsilon^2 \cdot
           \frac{\sigma_v}{\sqrt{\sigma_u^2 + \sigma_\varepsilon^2}}.
\end{align*}
For LVR the second factor $V_{\mathrm{AMM}}(q_t)/8$ is independent
of $\sigma$; for the privacy subsidy the second factor
$\sigma_v / \sqrt{\sigma_u^2 + \sigma_\varepsilon^2}$ is itself a
function of $\sigma_\varepsilon$. The factorization and the
solvency criterion $\int\!f \ge \int\!\ell$ are \emph{global}.
Only the scaling interpretation \emph{``rate is quadratic in the
noise driver''} is asymptotic: exact for LVR over the full
parameter range, valid for the privacy subsidy in the small-noise
regime $\sigma_\varepsilon \ll \sigma_u$ (where the second factor
$\approx \sigma_v/\sigma_u$ is approximately noise-independent
and $\ell^{\mathrm{priv}} \sim \sigma_v\,\sigma_\varepsilon^2/\sigma_u$),
and degrading to linear in $\sigma_\varepsilon$ in the large-noise
regime $\sigma_\varepsilon \gg \sigma_u$
($\ell^{\mathrm{priv}} \sim \sigma_v\,\sigma_\varepsilon$). The
cumulative welfare in each framework is the time-integral
$W = \int_0^T \ell(t)\,dt$; for the privacy subsidy the rate is
constant in $t$, so $W = \ell^{\mathrm{priv}} \cdot T$, enabling
the closed-form solvency criterion of~\Cref{sec:app-fee} at any
noise level.
\end{proposition}

\subsection{Why the bridge matters}

Milionis et al.~\cite{milionis2022lvr} establish LVR as the
foundational welfare quantity for CFMM design: a CFMM is solvent
over a trading horizon only if cumulative fee revenue exceeds
cumulative LVR. The present paper adds the order-flow observation
analog: a privacy-aggregated exchange (shielded AMM, MPC matching
engine, sealed-bid auction with noisy revelation) is solvent only
if cumulative fee revenue exceeds the cumulative privacy subsidy.
LVR addresses mismatched price observation; the privacy subsidy
addresses mismatched flow observation. Both enter the break-even
fee inequality for committed-curve exchanges, one governing price
observation and one flow observation.

%-------------------------------------------------------------
\section{Applications}
\label{sec:applications}

\subsection{Shielded AMM with Gaussian noise injection}
\label{sec:app-shielded}

A shielded AMM that publishes the post-trade pool state but
adds independent Gaussian privacy noise to each block's net
order flow is approximated by our model under a discrete-block
interpretation; as in the single-period
companion~\cite{nakamura2026privacysubsidy}, this is an idealized
normative benchmark, not a literal model of a deployed
constant-function AMM (which prices off its bonding curve, not a
posterior mean over a latent value). The protocol commits to a price-impact
coefficient $\lambda$ in the smart-contract code, traders submit
shielded swaps in $N$ blocks of length $1/N$ each, and the
published flow on each block is the true flow plus an
independent Gaussian privacy increment of variance
$\sigma_\varepsilon^2/N$ (so that the cumulative
privacy-channel variance matches the Brownian intensity
$\sigma_\varepsilon^2$ of our continuous-time model). The
DP-feasibility side of this construction is developed in detail
by Chitra, Angeris, and Evans~\cite{chitra2022dpcfmm}, who
introduce a \emph{Uniform Random Execution} mechanism achieving
$(\epsilon,\delta)$-DP in constant-function market makers and
characterise the achievable privacy parameter in terms of the
AMM curve's curvature and the trade count. Our cumulative-subsidy
result is the welfare-cost complement to their feasibility
analysis: given that a DP layer of intensity $\sigma_\varepsilon$
is realised in the implementation, the LP pool pays the
cumulative subsidy of~\Cref{thm:subsidy} over the trading
horizon.

\paragraph{DP mapping and a continuous-time obstruction.}
\begin{sloppypar}
We write $\epsilon$ for the differential-privacy budget, distinct
from the noise glyph in $\sigma_\varepsilon$. For a single block,
the Gaussian mechanism provides
$(\epsilon_{\mathrm{block}},\delta)$-DP at noise std
$\sigma_{\mathrm{block}}$ via
$\sigma_{\mathrm{block}} = \Delta\sqrt{2\log(1.25/\delta)}/\epsilon_{\mathrm{block}}$,
where $\Delta$ is the unit-trade sensitivity. Setting
$\sigma_{\mathrm{block}} = \sigma_\varepsilon/\sqrt{N}$ gives the
per-block privacy budget
$\epsilon_{\mathrm{block}} \propto \sqrt{N}/\sigma_\varepsilon$.
The joint $(\epsilon, \delta)$ budget over $N$ blocks then
scales as
$N \cdot \epsilon_{\mathrm{block}} = O(N^{3/2}/\sigma_\varepsilon)$
and $N \cdot \delta$ under basic composition, and as
$O(N/\sigma_\varepsilon)$ under advanced (or Rényi-DP)
composition with tighter $\delta$ control. Both regimes diverge
as $N \to \infty$ at fixed $\sigma_\varepsilon$: a continuous-time
Brownian privacy channel of fixed intensity
$\sigma_\varepsilon > 0$ is incompatible with a finite joint DP
guarantee.
\end{sloppypar}

\paragraph{Discrete deployment is the right operational
interpretation.}
The continuous-time model of~\Cref{thm:subsidy} should therefore
be read as the diffusion limit of a finite-$N$ discrete
deployment, not as a model that itself satisfies DP at finite
budget. For any chosen deployment $(N,
\epsilon_{\mathrm{joint}})$, the per-block budget is
$\epsilon_{\mathrm{block}} = \epsilon_{\mathrm{joint}}/N$ (basic)
or $\epsilon_{\mathrm{joint}}/\sqrt{N}$ (advanced), the per-block
noise std follows from the Gaussian mechanism, and the
corresponding $\sigma_\varepsilon^2 = N \sigma_{\mathrm{block}}^2$
substitutes into~\Cref{thm:subsidy} to give the break-even fee.
The continuous-time formula remains valid as the $N \to \infty$
limit at appropriately scaled $\epsilon_{\mathrm{joint}}$ (so
that $\sigma_\varepsilon$ stays bounded as the discretisation
refines), but the DP budget is consumed at the chosen $N$.

\subsection{MPC matching engines}
\label{sec:app-mpc}

In MPC-based matching engines, the protocol observes an order
flow signal that has been intentionally coarsened by the secure
multiparty computation reveal step. If the MPC protocol injects
Gaussian noise of variance $\sigma_\varepsilon^2$ during reveal,
the engine's pricing rule operates on the noisy signal exactly
as in our model. The single-period
analysis~\cite{nakamura2026privacysubsidy} already exhibits the
basic privacy-vs-subsidy trade-off; the continuous-time setting
adds one substantive observation specific to the multi-step
protocol design.

The protocol may attempt to choose $\sigma_\varepsilon(t)$ as a
deterministic \emph{function} of time -- for instance, weaker
privacy near the horizon when most information has been
impounded -- in the hope of reducing the cumulative subsidy at
fixed average privacy budget. Solving~\Cref{thm:equilibrium} with
time-varying $\sigma_\varepsilon(t)$ shows that this attempt
\emph{cannot succeed}. The HJB analysis still forces $\alpha$ and
$\lambda$ to be constants in $t$ (Step~2 of the proof is
unchanged when $\sigma_\varepsilon$ depends on $t$), the Bayes
identity then gives $\beta_t \Sigma(t) = \lambda\,(\sigma_u^2 +
\sigma_\varepsilon(t)^2)$, and the Riccati becomes
\(
  d\Sigma/dt = -\,\lambda^2\,(\sigma_u^2 + \sigma_\varepsilon(t)^2).
\)
Integrating and imposing $\Sigma(1) = 0$ yields
$\lambda^2 = \sigma_v^2 / (\sigma_u^2 + \langle\sigma_\varepsilon^2\rangle)$
where $\langle\sigma_\varepsilon^2\rangle := \int_0^1
\sigma_\varepsilon(t)^2\,dt$ is the time-averaged variance.
Consequently $\ell^{\mathrm{priv}}(t) = \lambda\,\sigma_\varepsilon(t)^2$
and
\[
  |\Pi_M| \;=\; \int_0^1 \lambda\,\sigma_\varepsilon(t)^2\,dt
          \;=\; \lambda\,\langle\sigma_\varepsilon^2\rangle
          \;=\; \frac{\sigma_v\,\langle\sigma_\varepsilon^2\rangle}
                     {\sqrt{\sigma_u^2 + \langle\sigma_\varepsilon^2\rangle}}.
\]

\noindent The cumulative subsidy depends on the privacy-noise
profile only through its time-averaged variance. Front-loading,
back-loading, or any other temporal arrangement that preserves
$\langle\sigma_\varepsilon^2\rangle$ produces the same subsidy
and the same break-even fee: the protocol cannot reduce the
subsidy by scheduling alone.

\subsection{Fee calibration}
\label{sec:app-fee}

\Cref{prop:lvr}'s break-even principle specializes to the
present setting as follows. Let $Q$ denote expected total
volume cleared over $[0,1]$. We assume $Q < \infty$, which
requires the discrete-deployment interpretation
of~\Cref{sec:app-shielded} ($N$ finite blocks); in the
continuous-time Brownian limit, total variation of $du_t$ is
infinite, so $Q$ diverges and a flat volume-proportional fee
must be interpreted block-wise rather than as a Stieltjes
integral. Under the discrete deployment, if the protocol charges
a flat proportional fee $f$ on each unit of volume, total fee
income over $[0,1]$ is $f \cdot Q$. Solvency of the liquidity
pool against the privacy subsidy requires $f \cdot Q \ge |\Pi_M|$,
i.e.,
\[
  f \;\ge\; \frac{\sigma_v\,\sigma_\varepsilon^2}
                 {Q\,\sqrt{\sigma_u^2 + \sigma_\varepsilon^2}}.
\]
This is the privacy analog of the LVR-derived break-even fee. The
companion single-period analysis~\cite{nakamura2026privacysubsidy}
gives half this value; \Cref{rem:kb-doubling} shows the factor of
two between the two is a consequence of the general Kyle--Back
single-period-to-continuous-time scaling, not a privacy-specific
phenomenon. The continuous-time fee is the appropriate
calibration for a continuous-time protocol; the single-period fee
remains correct for a one-shot batched auction.

%-------------------------------------------------------------
\section{Discussion and Future Work}
\label{sec:discussion}

\subsection{Non-Gaussian privacy noise}

The Gaussian privacy channel admits a clean closed-form because
the joint distribution remains in the exponential family and the
Bayesian Kalman gain is linear. For $\epsilon$-differential
privacy with Laplace noise, the analog of~\Cref{thm:equilibrium}
loses linearity and the pricing rule is no longer affine in the
observed flow. We conjecture that the privacy-subsidy rate
qualitatively retains the form
$O(\sigma_\varepsilon^2)$ for small $\sigma_\varepsilon$ but
deviates for large noise; quantitative analysis is left for
future work.

\subsection{Jump processes and FBSDE techniques}

The forward--backward stochastic differential equation framework
of~\cite{cetin2017financial} extends Kyle--Back to non-Gaussian
asset-value processes including jumps. Privacy noise on order
flow combines naturally with this framework, yielding a coupled
FBSDE in which the backward component is the value function of
a partially-observed control problem. We expect
\Cref{thm:subsidy}'s structure -- $|\Pi_M|$ scaling as
$\sigma_\varepsilon^2 / \sqrt{\sigma_u^2 + \sigma_\varepsilon^2}$
in the noise parameters -- to persist qualitatively, with the
closed form replaced by an integral over the equilibrium price
trajectory.

\subsection{Multiple informed traders}

Foster--Viswanathan~\cite{foster1996strategic} extends Kyle to
multiple informed traders forecasting each other's forecasts. The
committed Bayesian-AMM extension is mechanically straightforward
but quantitatively non-trivial: the price impact $\lambda$
depends on the number of insiders and their correlation
structure, and the privacy subsidy decomposes across insiders
according to their relative information contributions. This
extension is most relevant to MEV-bot environments where
multiple competing search agents observe correlated signals.

\subsection{Continuous-time Glosten--Milgrom analog}

Paper~B~\cite{nakamura2026glostenmilgrom} establishes the
single-period privacy subsidy in the discrete-value
Glosten--Milgrom~\cite{glosten1985} model with binary flip-noise
on the direction signal. A continuous-time extension -- the
discrete-value analog of the present paper -- is open and
non-trivial. We sketch the technical landscape.

The natural continuous-time embedding replaces the
single-shot binary trade with a Poisson stream: order arrivals
form a marked Poisson process whose marks are signed trade
directions $\xi \in \{+1, -1\}$. Each arrival is independently
informed (with probability $\mu$) or noise (with probability
$1-\mu$); an informed trader buys when $v = v_H$ and sells when
$v = v_L$, while a noise trader buys or sells uniformly at
random. The marginal probability of a buy given $v = v_H$ is
therefore $\Pr(\xi = +1 \mid v = v_H) = \mu + (1-\mu)/2
= 1/2 + \mu/2$, recovering the discrete-time setup
of~\cite{nakamura2026glostenmilgrom}. The privacy channel is a
binary flip applied independently to each arrival with flip rate
$\eta$. The market maker observes the noisy directional stream
and Bayes-updates its belief
$\pi_t := \Pr(v = v_H \mid \mathcal{F}_t^{\mathrm{MM}})$.

The technical obstacle is that $\pi_t$ is the natural state
variable but its dynamics are non-Gaussian and inherently jump
driven: each arrival induces a discrete Bayes update of $\pi_t$
by a multiplicative likelihood ratio. The closed-form
Kalman-gain reduction of~\Cref{thm:equilibrium} does not apply.
Instead, the posterior follows a piecewise-deterministic
Markov process whose generator combines drift (between
arrivals) with jumps (at arrivals). The equilibrium bid-ask spread is the analog of~$\lambda$ but is
itself a function of the current belief $\pi_t$, so the per-trade
subsidy is in general state-dependent. The cumulative subsidy is
thus the integral $\mathbb{E}[\int_0^1 \mu\eta\Delta(\pi_t)\,d N_t]$,
where $N_t$ is the trade-arrival process and
$\mu\eta\Delta(\pi_t)$ is the per-trade subsidy at the current
belief. In the uninformative-prior approximation $\pi_t \equiv 1/2$
the integrand collapses to the constant
$\mu\eta\Delta$ of~\cite{nakamura2026glostenmilgrom} and the
cumulative subsidy reduces to $\mu\eta\Delta \cdot \Lambda$ for
total expected trade count $\Lambda$; the general state-dependent
case requires the full PDMP analysis and is left for future work.

%-------------------------------------------------------------
\bibliographystyle{splncs04}
\bibliography{references}

\end{document}